\documentclass[aps,pra,reprint,twoside]{revtex4-2}
\usepackage[english]{babel}

\usepackage{graphicx}
\usepackage{dcolumn}
\usepackage{bm}
\usepackage{physics}
\usepackage{amsmath, amssymb, amsfonts, amsthm, dsfont}
\usepackage{svg}
\usepackage{color}
\usepackage{url}
\usepackage[hidelinks]{hyperref}
\usepackage[capitalise]{cleveref}
\usepackage{microtype}
\usepackage[normalem]{ulem}
\usepackage{stmaryrd}
\usepackage{multirow}
\usepackage{booktabs}

\definecolor{darkblue}{RGB}{50,10,180}
\definecolor{orangef}{RGB}{210,100,20}

\newcommand{\id}{\mathds{1}}

\newtheorem{Definition}{Definition}

\newtheorem{Proposition}{Proposition}
\newtheorem*{Example}{Ex}

\begin{document}

\title{
Extremal Tsirelson inequalities}

\author{Victor Barizien}
\affiliation{Université Paris Saclay, CEA, CNRS, Institut de physique théorique, 91191 Gif-sur-Yvette, France}
\author{Jean-Daniel Bancal}
\affiliation{Université Paris Saclay, CEA, CNRS, Institut de physique théorique, 91191 Gif-sur-Yvette, France}

\date{\today}

\begin{abstract}
It is well-known that the set of statistics that can be observed in a Bell-type experiment is limited by quantum theory. Unfortunately, tools are missing to identify the precise boundary of this set. Here, we propose to study the set of quantum statistics from a dual perspective. By considering all Bell expressions saturated by a given realization, we show that the CHSH expression can be decomposed in terms of extremal Tsirelson inequalities that we identify. This brings novel insight into the geometry of the quantum set in the (2,2,2) scenario. Furthermore, this allows us to identify all the Bell expressions that are able to self-test the Tsirelson realization.
\end{abstract}

\maketitle

\textit{Introduction --} Quantum physics predicts the existence of statistics in a Bell-type experiment which are non-local in the sense that they can violate a Bell inequality~\cite{Brunner14}. The observation of this striking physical property has raised awareness on the importance of the statistical distributions that can be observed in a Bell scenario upon measurement of a quantum system, which form the set of quantum behaviors -- or simply the \emph{quantum set} $\mathcal{Q}$. By defining statistics that are experimentally realizable, this set plays a central role in quantum information science, with applications ranging from foundational questions to device-independent information protocols~\cite{Hardy93,Mayers04,Pironio10a,Nadlinger22,Zhang22a,Liu22a}. 
Indeed, every quantum realization, consisting of a density matrix and local measurements on some Hilbert spaces of arbitrary dimension, generates statistics according to Born's rule which belong to the quantum set. Given a statistical distribution, it is however generally a difficult question to determine which quantum realization may generate it, or even whether the distribution belongs to the quantum set in the first place.

The Navascués-Pironio-Acín (NPA) hierarchy of semidefinite programming offers a tool to tackle this question in terms of a family of statistical sets converging to the quantum set from the outside~\cite{Navascues07,Navascues08}. However, this construction involves additional variables whose value is a priori unknown, and the level of the hierarchy which must be reached in order to provide a definitive answer is unknown even in the simplest Bell scenario. 
Another approach is motivated by the convexity of the quantum set and involves decomposing quantum behaviors in terms of extremal points. 
Despite much effort, a complete description of extremal quantum behaviors remains to be discovered~\cite{Masanes03,Ishizaka18,Mikosnuszkiewicz23}. Here, we take a different path, also enabled by $\mathcal{Q}$'s convexity property, which consists in studying its dual set $\mathcal{Q}^*$.

The dual of the quantum set encodes the Tsirelson bounds of all possible Bell inequalities~\cite{Tsirelson87,Wehner06}. Describing it is therefore as challenging as finding the quantum bound of an arbitrary Bell inequality, but also as important, since quantum bounds play a key role in many quantum information results~\cite{Pawlowski09,Navascues10,Supic20}. The duality perspective already brought insight into the local and no-signaling sets, which are the two other major sets of interest in Bell-type experiments. Namely, for Bell scenarios with binary inputs and outputs it was shown that the local set, describing statistical distributions compatible with a local hidden variable model, is dual to the no-signaling set, whose behaviors are only limited by the condition that the parties cannot learn each other's inputs~\cite{Fritz12}. In other words, every extremal (or `tight') Bell inequality in this scenario is in one-to-one correspondence with an extremal point of the no-signaling polytope.

Unlike its local and no-signaling counterparts, the quantum set is not a polytope and little is known about its dual picture. In the simplest Bell scenario exhibiting the non-local property of quantum physics, with 2 parties, 2 inputs and 2 outputs, the quantum set belongs to a space of dimension 8. A first result concerns the subset $\mathcal{Q}_\text{c}$ with uniformly random marginal statistics, corresponding to a subspace of dimension 4. It was recently shown that this subset is self-dual, i.e.~$\mathcal{Q}_\text{c}\cong\mathcal{Q}_\text{c}^*$~\cite{Le23,Fritz12}. This striking property sets the quantum set apart from both the local and the no-signaling sets. 
In fact, the analytical descriptions of $\mathcal{Q}_\text{c}$ and $\mathcal{Q}_\text{c}^*$ are fully known within this subspace: a first explicit description of the quantum set in the subspace of vanishing marginals was provided in~\cite{Tsirelson87,Landau88,Masanes03}; see also~\cite{Le23,Barizien23,Wooltorton23} for explicit descriptions of its (isomorphic) dual. 


Here, we study the dual of the quantum set in the full 8-dimensional space. Specifically, we determine analytically all elements of the dual which are related to the Tsirelson point, the unique quantum point maximally violating the Clauser-Horne-Shimony-Holt (CHSH) inequality~\cite{Clauser69}. This allows us to describe for the first time a complete face of the dual quantum set $\mathcal{Q}^*$'s boundary. In turn, this provides a tight first order description of the quantum set around this maximally nonlocal point.



\textit{Dual of the quantum set --} In a bipartite Bell experiment, two parties obtain outcomes $a$ and $b$ upon performing measurements $x$ and $y$ respectively. A behavior $P(ab|xy)$ in this scenario belongs to the quantum set $\mathcal{Q}$ iff there exists Hilbert spaces $\mathcal{H}_A$, $\mathcal{H}_B$, a density matrix $\rho\geq 0$ with $\tr\rho=1$ acting on their tensor product $\mathcal{H}_A\otimes\mathcal{H}_B$, and POVMs $M_{a|x}, N_{b|y}\succeq 0$ with $\sum_a M_{a|x}=\openone$, $\sum_b N_{b|y}=\openone$ such that $P(ab|xy)=\tr(M_{a|x}\otimes N_{b|y} \rho)$. Since the dimensions of the Hilbert spaces $\mathcal{H}_A$ and $\mathcal{H}_B$ are not bounded, any convex mixture $\lambda \bm P_1 + (1-\lambda) \bm P_2$ with $\lambda\in[0,1]$ of two behaviors $\bm P_1$ and $\bm P_2$ in $\mathcal{Q}$ can be obtained by combining the two corresponding realizations into larger Hilbert spaces, and therefore $\mathcal{Q}$ is convex~\cite{Slofstra19}.



By taking into account the normalization and no-signaling conditions, the 16 conditional probabilities $P(ab|xy)$ can be expressed simply in terms of 8 linearly independent ones~\cite{Collins04,Rosset20}, which can be represented by the corresponding table of correlators
\begin{equation}
    \bm P = \begin{array}{c|c|c}
         1 & \langle B_0\rangle & \langle B_1\rangle \\
         \hline
         \langle A_0\rangle & \langle A_0 B_0\rangle & \langle A_0 B_1\rangle \\
         \hline
         \langle A_1\rangle & \langle A_1 B_0\rangle & \langle A_1 B_1\rangle
    \end{array}.
\end{equation}
Here, $A_x=M_{0|x}-M_{1|x}$ and $B_y=N_{0|y}-N_{1|y}$ are observables with $\pm 1$ eigenvalues. This allows one to define the dual of the quantum set $\mathcal{Q}^*$ in $\mathbb{R}^8$ as the set of all Bell expressions $\beta$ whose quantum maximum is smaller than a constant, e.g.~1 (see also \cref{sec:indices}):
\begin{equation}
\mathcal{Q}^* = \{\beta\in\mathbb{R}^8: \beta\cdot\bm P \leq 1,\ \forall \bm P\in\mathcal{Q}\}.
\end{equation}
Since the double dual of a cone is the closure of the initial cone, the description of $\mathcal{Q}^*$ is equivalent to the description of $\mathcal{Q}$ itself, and any insight on $\mathcal{Q}^*$ is an insight on $\mathcal{Q}$ as well.

Note that all elements of $\mathcal{Q}^*$ with
\begin{equation}
\beta\cdot \bm P=1
\label{eq:condition}
\end{equation}
for some $\bm P\in\mathcal{Q}$ are Bell expressions defining supporting hyperplanes of the quantum set $\mathcal{Q}$. Such inequalities provide a description of the quantum set around the point $\bm P$ to first order. More generally, the dual of the quantum set $\mathcal{Q}^*$ being convex, it admits extremal points, which are of particular interest. An example of an extremal point of $\mathcal{Q}^*$ is the positivity constraint $P(ab|xy)\geq0$. However, this point is not specific to the quantum dual as it is shared with every other physically meaningful dual, including the local and no-signaling duals. In the remainder of this manuscript, we are going to identify non-trivial extremal points of $\mathcal{Q}^*$.



\textit{The Tsirelson behavior --} The Tsirelson point is given by the following table of correlators
\begin{equation} \label{eq:point}
    \bm P_{T} = \begin{array}{c|c|c}
         1 & 0 & 0 \\
         \hline
         0 & \frac{1}{\sqrt{2}} & \frac{1}{\sqrt{2}} \\
         \hline
         0 & \frac{1}{\sqrt{2}} & -\frac{1}{\sqrt{2}}
    \end{array}.
\end{equation}
This point is particularly remarkable because it is the only point in $\mathcal{Q}$ that achieves the maximal quantum value of the CHSH Bell inequality $\langle \beta_\text{CHSH} \rangle\leq 2$~\cite{Cirelson80}. Furthermore, this point is extremal in $\mathcal{Q}$, and it can only be realized by performing complementary measurements on a maximally entangled state, i.e.~the point $\bm P_T$ self-tests the quantum realization~\cite{Supic20}:
\begin{equation}
\label{eq:realization}
\begin{split}
& \ket{\phi^+} = \frac{1}{\sqrt{2}}(\ket{00}+\ket{11}),\\
& A_x = \frac{Z_A + (-1)^x X_A}{\sqrt{2}}, \quad B_0 = Z_B, \ B_1 = X_B.
\end{split}
\end{equation}

As shown in \cref{eq:point}, the Tsirelson point has vanishing marginals, and when considering $\mathcal{Q}_\text{c}$, the quantum set within the subspace of vanishing marginals, i.e.~with $\langle A_x \rangle = \langle B_y \rangle=0$, it is known that this point is only exposed by the CHSH inequality. In other words, the hyperplane $\langle\beta_{\text{CHSH}}\rangle = 2\sqrt{2}$ is the only linear function of $\{\langle A_x B_y\rangle\}_{x,y}$ such that $H\cap \mathcal{Q}_\text{c} = \{\bm P_{T}\}$. Recent numerical results suggest however that this may not be the case outside the subspace of vanishing marginals~\cite{Goh18}. In the following, we identify all Bell expressions that are maximized by the Tsirelson point.

\textit{Bell expressions for the Tsirelson point --} In general, all possible Bell expressions can be parametrized by eight real coefficients $a_x, b_y, c_{xy}$ for $x,y\in\{0,1\}$ and be written in terms of formal polynomials~\cite{Barizien23} as
\begin{equation}
\begin{split}
    \beta & = a_0 A_0 + a_1 A_1 + b_0 B_0 + b_1 B_1 \\
    &+ c_{00} A_0B_0 + c_{10} A_1B_0 + c_{01} A_0B_1 + c_{11} A_1B_1.
\end{split}
\end{equation}
In order to find restrictive conditions that ensure that $\beta$ has quantum bound $1$ and verifies~\cref{eq:condition} for the Tsirelson point $\bm P_T$, we make use of the variational method~\cite{Barizien23,Pal14,Sekatski18} and consider the Bell operator corresponding to these polynomials for the choice of measurements of \cref{eq:realization}. In general, this operator is given by
\begin{equation}
\begin{split}
    \hat S = & p_1 Z_A + p_2 X_A + p_3 Z_B + p_4 X_B \\
    & \ + p_5 Z_AZ_B + p_6 X_AX_B + p_7 Z_AX_B + p_8 X_AZ_B,
\end{split}
\end{equation}
where parameters $p_r$ are linear combinations of parameters $a_x,b_y,c_{xy}$.

For any given state, the value of the Bell inequality is then given by $\bra{\psi} \hat S \ket{\psi}$. Since $\hat S$ is a real hermitian operator, its eigenvalues are real and $\beta \in \mathcal{Q}^*$ imposes that all of its eigenvalues are smaller than~$1$. The condition \cref{eq:condition} for $\bm P_T$ then implies that $\ket{\phi^+}$ is an eigenstate of $\hat S$ of eigenvalue~$1$, i.e.~$\hat S \ket{\phi^+}=\ket{\phi^+}$, which grants the following conditions on the parameters $p_r$:
\begin{equation}
    \left\{ \begin{split}
        & p_1 + p_3 = 0, \\
        & p_2 + p_4 = 0, \\
    \end{split}\right., \quad \left\{ \begin{split}
        & p_5 + p_6 = 1, \\
        & p_7 - p_8 = 0.
    \end{split}\right.
\end{equation}
Taking these equations into account, we can now rewrite a new parametrization:
\begin{equation}
\begin{split}
    \beta & = r_0 \left(\frac{A_0+A_1}{\sqrt{2}} - B_0\right) + r_1 \left(\frac{A_0-A_1}{\sqrt{2}} - B_1\right)\\
    & \qquad + r_2 \left(\frac{A_0+A_1}{\sqrt{2}}B_1 + \frac{A_0-A_1}{\sqrt{2}}B_0\right) \\
    & \qquad + \lambda \frac{A_0+A_1}{\sqrt{2}} B_0 + (1-\lambda) \frac{A_0-A_1}{\sqrt{2}} B_1,
\end{split}
\end{equation}
where $r_0, r_1, r_2, \lambda \in \mathbb{R}$ are the remaining free parameters.

\textit{Perturbative restriction --} Now, remember that we require that no other quantum point gives a value larger than~$1$ for this Bell expression. In particular, the function $\beta \cdot \bm P_{\theta,a_x,b_y}$ should admit a local maxima at $\bm P_{T}$, where
\begin{equation}
    \bm P_{\theta,a_x,b_y} = \begin{array}{c|c|c}
         1 & c_{2\theta}c_{b_1} & c_{2\theta}c_{b_2} \\
         \hline
         c_{2\theta}c_{a_1} & \multicolumn{2}{c}{\multirow{2}{*}{$c_{a_x}c_{b_y}+s_{2\theta}s_{a_x}s_{b_y}$}} \\
         \cmidrule{1-1}
         c_{2\theta}c_{a_2} & \multicolumn{2}{c}{}
    \end{array}, \quad \theta,a_x,b_y \in \mathbb{R}
\end{equation}
are the statistics resulting from measuring the two-qubit state $\ket{\phi_\theta}=c_\theta\ket{00}+s_\theta\ket{11}$ in the $Z-X$ plane, and $c_\varphi := \cos(\varphi)$, $s_\varphi := \sin(\varphi)$ denote the cosine and sine functions. This condition gives a set of five linear equations:  
\begin{equation}
\begin{split}
    0 = \beta \cdot \frac{\partial \bm P_{\theta,a_x,b_y}}{\partial \theta}= \beta \cdot \frac{\partial \bm P_{\theta,a_x,b_y}}{\partial a_x} = \beta \cdot \frac{\partial \bm P_{\theta,a_x,b_y}}{\partial b_y}
\end{split}
\end{equation}
which reduce to
\begin{equation}
\lambda = 1/2, \quad r_2 = 0.
\end{equation}
The search space for the Bell inequalities is thus reduced to
\begin{equation}
\label{eq:bellfamily}
\begin{split}
    \beta_{r_0,r_1} = r_0& \left(\frac{A_0+A_1}{\sqrt{2}} - B_0\right) + r_1 \left(\frac{A_0-A_1}{\sqrt{2}} - B_1\right)  \\
    & + \frac{1}{2\sqrt{2}}\beta_{\text{CHSH}}, \quad r_0, r_1 \in \mathbb{R},
\end{split}
\end{equation}
where $\beta_{\text{CHSH}}= (A_0+A_1)B_0 + (A_0-A_1)B_1$. As shown in \cref{sec:orderTwo}, further order perturbations can be considered to reduce the range of the parameters $r_0,r_1$, but it turns out to be more restrictive at this stage to eliminate parameters based on the local bound of the Bell expressions. Indeed, any $\beta_{r_0,r_1}$ with a local bound larger than~$1$ also admits a quantum value larger than~$1$ (and hence larger than the value provided by measuring the $\ket{\phi^+}$ state).

\textit{Local bounds --} Since the convex combination of two Bell expressions with a local bound smaller than~$1$ also has a local bound smaller than~$1$, the set of Bell expressions $\beta_{r_0,r_1}$ with a local bound smaller than~$1$ forms a convex region of the $r_0,r_1$ plane. Furthermore, the local maxima of a Bell expression is reached at one of the 16 extremal points of the local polytope. These points are given by
\begin{equation}
    \bm L_{ijkl} = \begin{array}{c|c|c}
         1 & i & j \\
         \hline
         k & ik & jk \\
         \hline
         l & il & jl
    \end{array}, \quad i,j,k,l \in \{-1,1\}.
\end{equation}
The convex region of expressions of the form \cref{eq:bellfamily} with a local maxima smaller than~$1$ is thus given by all points $(r_0,r_1)$ satisfying the conditions $\beta_{r_0,r_1}\cdot \bm L_{ijkl} \leq 1$.
The intersection of these half planes defines a polytope, namely a regular octagon, whose eight summit are given by (see \cref{fig:dual})
\begin{equation}
    \{\left(1-\frac{1}{\sqrt{2}}\right)R_{\frac{\pi}{4}}^k(1,0), \ k\in \llbracket 0,7 \rrbracket \},
\end{equation}
where $R_{\frac{\pi}{4}}$ is the rotation of angle $\pi/4$ in the $(r_0,r_1)$ plane. Any Bell expression outside this octagon has a local bound larger than~$1$, and thus is not maximized by $\bm P_T$.

\textit{Quantum bounds --} Having excluded a range of Bell expressions $\beta_{r_0,r_1}$ from their local bound, we need to compute the Tsirelson bound of the expressions inside the octagon. To formalize this problem, let us consider the algebra of quantum operators $\mathcal{R}$ made of arbitrary products of $A_x$ and $B_y$. The only generating rules of this algebra are that, for all $x,y$, $A_x^2=B_y^2=\id$ and $[A_x,B_y]=0$. We know that a sufficient condition to have $\beta \preceq 1$ is that $1-\beta$ is a sum of squares (SOS) in $\mathcal{R}$:
\begin{equation}
\label{SOS}
    1-\beta = \sum_s O_s^\dagger O_s, \quad O_s\in\mathcal{R}.
\end{equation}
This is known as an SOS relaxation~\cite{Ioannou22,Liang07}. However, the search space $\mathcal{R}$ is of infinite dimension. Since an SOS decomposition of the form of Eq.~(\ref{SOS}) where $O_s$ are restricted to a set subspace $\mathcal{T}\subset\mathcal{R}$ still provides a valid bound, a common approach to tackle this problem consists in considering operators $O_s$ within a chosen relaxation level $\mathcal{T}$, such as the set of all polynomials of $A_x$ and $B_y$ of a given degree. But even this quickly results in a large problem.

To further reduce the SOS search space, let us identify a relevant subspace of $\mathcal{T}$ in which the operators $O_s$ should be chosen. Let's consider the case where a SOS decomposition exists. For the implementation of the Tsirelson realization, this would imply
\begin{equation}
    0 = \sum_s \bra{\phi^+} \Bar{O_s}^\dagger \Bar{O_s} \ket{\phi^+} = \sum_s ||\Bar{O_s} \ket{\phi^+}||^2,
\end{equation}
where $\Bar{O_s}$ is the specific implementation of the operator $O_s$ using measurements of~\cref{eq:realization}. Since all terms on the right-hand side of the above equation are positive, this implies that for all $s$, $\Bar{O_s} \ket{\phi^+}=0$, i.e. that all $\Bar{O_s}$ are nullifying operators of $\ket{\phi^+}$~\cite{Bamps15, Barizien23}. This condition restricts the operators $O_s$ to
\begin{equation}
    \mathcal{A}_\mathcal{T} = \{ O \in \mathcal{T} : \Bar{O}\ket{\phi^+} = 0\}.
\end{equation}

For a finite relaxation $\mathcal{T}$, let's consider a generating sequence $\{N_s\}_s$ of $\mathcal{A}_\mathcal{T}$ and denote by $\vec N$ the vector of elements $N_s$. All elements in $\mathcal{A}_\mathcal{T}$ can thus be written as $\Vec{w}\cdot \vec N$ where $\vec w$ is a real vector. A valid SOS decomposition in $\mathcal{T}$ can then be written as
\begin{equation}
    1-\beta = \sum_s O_s^\dagger O_s = \vec N^\dagger \sum_s \Vec{w}_s^\dagger \Vec{w}_s \vec N = \vec N^\dagger \cdot W \cdot \vec N,
\end{equation}
where $W = \sum_s \Vec{w}_s^\dagger \Vec{w}_s$ is a positive matrix. We see that the problem of obtaining an SOS decomposition, \cref{SOS}, reduces to finding whether there exists a positive matrix $W \succeq 0$ such that
\begin{equation}\label{eq:1minusBeta}
1-\beta = \vec N^\dagger \cdot W \cdot \vec N.
\end{equation}
If such a matrix $W$ can be found, we say it is a certificate of the inequality $\beta$.


Different relaxations $\mathcal{T}$ could be considered here. The first order relaxation $\mathcal{T}_{1+A+B} = \{ 1, A_0,A_1,B_0,B_1\}$ only gives a certificate for the CHSH inequality. The relaxation at the almost quantum level~\cite{Navascues15}, $\mathcal{T}_{1+A+B+AB} = \mathcal{T}_1 \cup \{A_x B_y, \ x,y \in \{0,1\} \}$ can also be computed numerically and gives a certificate for a disk in the $(r_0,r_1)$ plane of center $(0,0)$ and radius $\frac{1}{4\sqrt{2}}$, c.f.~\cref{sec:orderTwo}. The next relaxation is given by
\begin{equation}
    \mathcal{T}_{1+A+B+AB + ABB'} = \mathcal{T}_{1+AB} \cup \{A_x B_y B_{y'}, y\neq y'\}.
\end{equation}
We can show analytically that a certificate can be found for the inequality $\beta_{1-1/\sqrt{2},0}$ at this level of relaxation (see \cref{sec:relaxation}). This ensures that the Bell expression
\begin{equation}
    \beta_T = \left(1-\frac{1}{\sqrt{2}}\right)\left( \frac{A_0+A_1}{\sqrt{2}} - B_0\right) + \frac{\beta_\text{CHSH}}{2\sqrt{2}}
\end{equation}
is maximized by the Tsirelson point. Since we already concluded that a larger value of $r_0$ admits a local value larger than~$1$, this bound for $r_1=0$ is the best we could have hoped for. One can check that this inequality is an extremal point of $\mathcal{Q}^*$, that it is exposed, in the sense that the quantum set admits a point which only saturates this Tsirelson inequality (see \cref{sec:indices}), and that it is only maximized by $3$ extremal points of $\mathcal{Q}$: $\bm P_T$ and two deterministic realizations (see \cref{sec:noOtherPoint}).

To analyze the rest of the octagon, we make use of some symmetries of the problem. Both the family of Bell expressions $\beta_{r_0,r_1}$ and the set $\mathcal{Q}$ are preserved by several discrete symmetries. One of them is described by the following action:
\begin{equation}
\label{symmetries}
        S : \ (r_0,r_1) \to R_{\frac{\pi}{4}}(r_0,r_1),\ \left\{\begin{split} & A_0 \to -B_1, \ A_1 \to -B_0, \\
        & B_0 \to -A_0, \ B_1 \to A_1.
    \end{split}\right.
\end{equation}
Due to this symmetry, the quantum bound of any Bell expression with parameters $(r_0,r_1) \in \mathbb{R}^2$ can be computed by looking at the bound of the inequality with parameters rotated by $\pi/4$. This ensures that the quantum bounds of the eight inequalities $S^k \cdot \beta_T$ of polar coordinates ${(1-\nobreak\frac{1}{\sqrt{2}}, k\frac{\pi}{4})}$ for $k\in \llbracket 0,7 \rrbracket$ are also~$1$. Therefore, the octagon is exactly the convex region of quantum bound equal to $1$. This completes the characterization of the slice $\beta_{r_0,r_1}$ (see \cref{fig:dual}).

\begin{figure}
    \centering
    \includegraphics[width=0.5\textwidth]{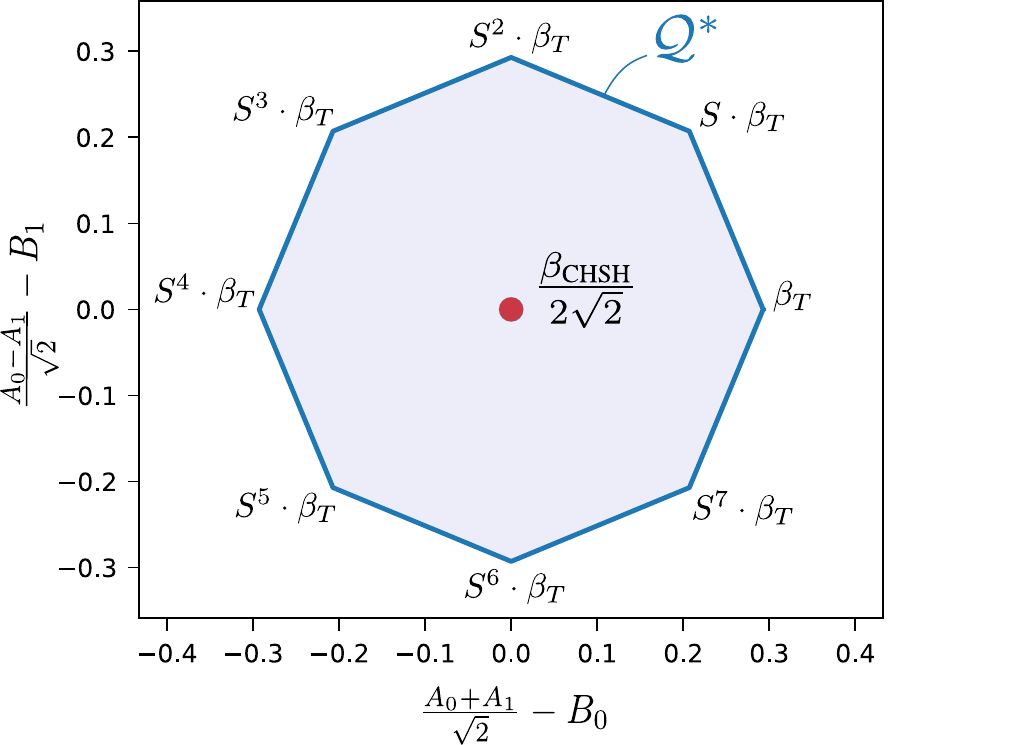}
    \caption{Face of $\mathcal{Q}^*$ in the two-dimensional affine slice defined by $\beta_{r_0,r_1}$ for real parameters $r_0,r_1$. The red point in the middle, the normalized CHSH expression, is non-extremal: it can be decomposed in terms of the eight summits of the octagon, which are extremal Tsirelson inequalities.}
    \label{fig:dual}
\end{figure}

Interestingly, the CHSH inequality lies in the middle of this dual face. Therefore, the CHSH inequality is not an extremal Tsirelson inequality. In particular, we can write it as the convex mixture
\begin{equation}\label{eq:CHSHdecomp}
\beta_\text{CHSH}=\frac{1}{2}\left(\beta_T + S^4 \cdot \beta_T\right).
\end{equation}
Note that this description is not unique because $\beta_\text{CHSH}$ lies on a face of $\mathcal{Q}^*$ of dimension 2.

From the point of view of the quantum set, this means that the Tsirelson point $\bm P_T$ is an exposed extremal point of $\mathcal{Q}$ with dimension pair $(0,2)$, i.e.~with a face dimension of 0 and a dual dimension of 2 (see \cref{sec:indices}). Furthermore, it is exposed by all the inequalities on the inside of the octagon (see \cref{fig:3dproj}). In fact, any Bell expression inside the octagon can be written as a convex combination of the CHSH expression and an expression $\beta_b$ on the border of the octagon: $\beta = p \beta_{\text{CHSH}}/2\sqrt{2} + (1-p) \beta_b$, where $p \in (0,1]$. If a point $\bm P$ verifies $\beta\cdot \bm P=1$, then it implies $\beta_\text{CHSH} \cdot \bm P= 2\sqrt{2}$ and the self-testing result of $\beta_\text{CHSH}$ implies that $\bm P = \bm P_T$. From the self-testing point of view, this means that any inequality inside the octagon self-tests the quantum realization of \cref{eq:realization} associated to the Tsirelson point. As far as inequalities on the border of the octagon are concerned, those are also maximized by local points and as such cannot provide a self-test of the realization.
\begin{figure}
    \centering
    \includegraphics[width=0.47\textwidth]{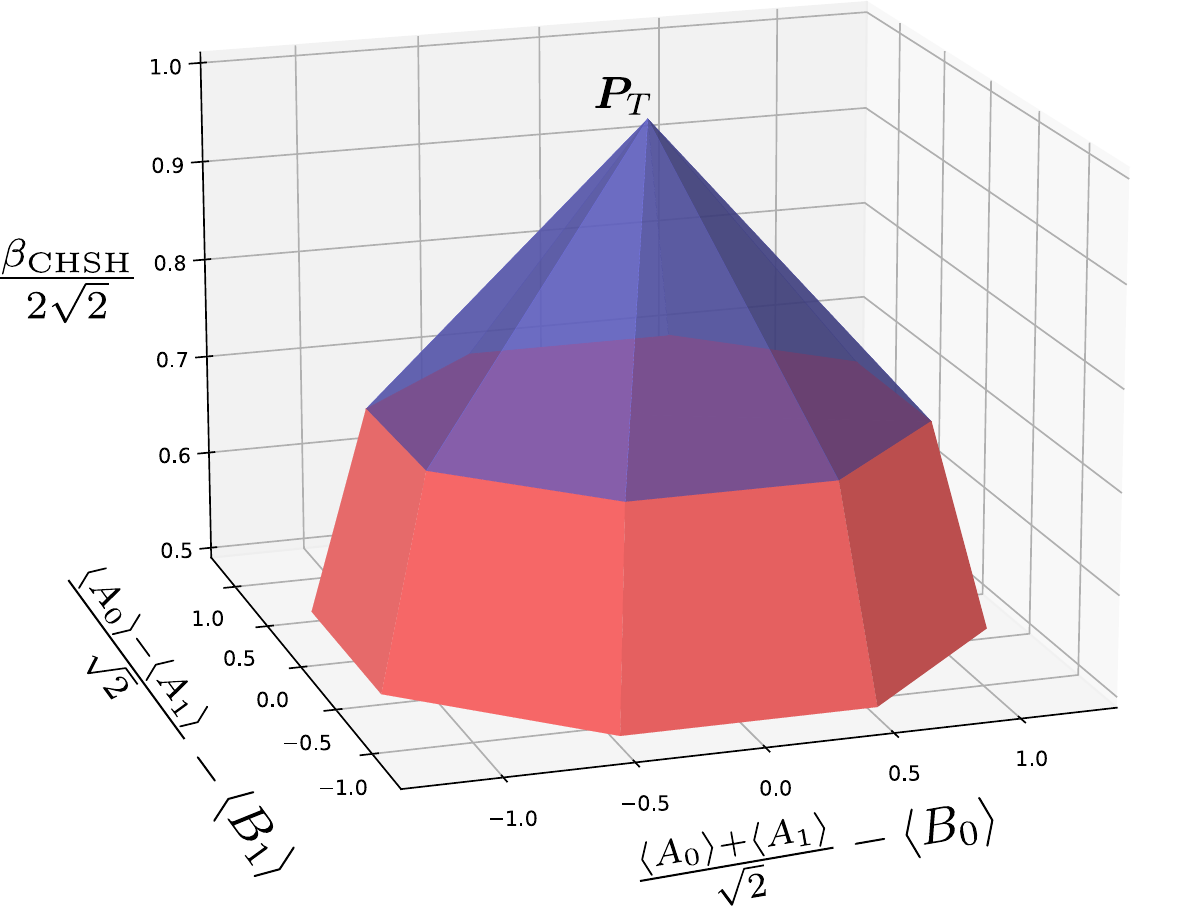}
    \caption{Three-dimensional projection of the local polytope (in red) and of the quantum set of correlations (red and blue). The only point reaching the z-value of $1$ is the Tsirelson realization. This point lies on top of an octagonal-based pyramid whose eight facets correspond to the inequalities $S^k \cdot \beta_T$.}
    \label{fig:3dproj}
\end{figure}

\textit{Conclusion --} In this paper, we studied the quantum set $\mathcal{Q}$ from a dual perspective. In particular, we derived constructively all the Bell expressions that the Tsirelson point $\bm P_T$ maximizes. This provides fresh insight on the geometry of the quantum set. In particular, we show analytically that $\bm P_T$ is an extremal point of $\mathcal{Q}$ of dual dimension 2 that lies at the top of a pyramid. We identify 8 new exposed extremal points of $\mathcal{Q}^*$, all of dual dimension 2 as well, thus fully describing a face of $\mathcal{Q}^*$ of dimension 2. In turn, this allows us to describe all the Bell expressions that are able to self-test the Tsirelson realization. It would be interesting to find out whether it is a generic property of extremal quantum statistics to have a non-zero dual dimension.

Our work also sheds light on the relation between $\mathcal{Q}$ and $ \mathcal{Q}^*$. In \cite{Fritz12}, a map was introduced to prove that $\mathcal{L}^* \cong \mathcal{NS}$ and it was proven that this map also sends the subset $\mathcal{Q}_c$ with uniformly random marginals statistics to its dual. However, since this map sends the extremal point $\bm P_T$ onto $\beta_\text{CHSH}$, which admits the decomposition~\cref{eq:CHSHdecomp}, it cannot be used to map $\mathcal{Q}$ to $\mathcal{Q}^*$. While other maps could be found, analysis such as ours might help proving that the quantum set $\mathcal{Q}$ is not self-dual, making $\mathcal{Q}^*$ a set of physical interest with a possibly very different geometrical structure.


\begin{acknowledgements}
We acknowledge funding by Commissariat à l’Energie Atomique et aux Energies Alternatives (CEA).
\end{acknowledgements}

\bibliographystyle{apsrev4-2}
\bibliography{biblio}{}

\appendix
\onecolumngrid

\section{Tools for characterization of convex sets}\label{sec:indices}
Let us recall some terminology useful to discuss convex sets in $\mathbb{R}^n$. We illustrate these definitions with results obtained in this work. For more formal details, see \cite{Rockafellar70}. In the following, $\mathcal{K}$ denotes a convex set over $\mathbb{R}^n$ and $\cdot$ the canonical scalar product. 

\begin{Definition}
    For any compact set $A \subset \mathbb{R}^n$, the \textbf{affine dimension} of $A$ is the vector space dimension of the set $A$ shifted by any of its elements: $\dim(\text{Vect}\langle A-\{x\}\rangle)$ for any $x\in A$.
\end{Definition}
\begin{Definition}
    A subset $F \subset \mathcal{K}$ is a \textbf{face} of $\mathcal{K}$, denoted $F \lhd \mathcal{K}$, iff:
    \begin{equation*}
        \forall y,z \in \mathcal{K}, \ \lambda\in(0,1), \ \lambda y + (1-\lambda)z \in F \Longrightarrow y,z \in F
    \end{equation*}
\end{Definition}
\noindent Note that a face $F\lhd \mathcal{K}$ is a convex set.
\begin{Definition}
    The \textbf{face dimension} of a face $F$, denoted $d_F$, is the affine dimension of $F$. 
\end{Definition}
\begin{Definition}
    A point $x$ of $\mathcal{K}$ is an \textbf{extremal point} if $\{x\}$ is a face of $\mathcal{K}$.
\end{Definition}
\noindent Note that extremal points are exactly the faces of dimension $0$.
\begin{Example}
    $\beta_T$ is an extremal point of $\mathcal{Q}^*$. Indeed, if their exist $g,h\in \mathcal{Q}^*$, $\lambda\in \mathcal{R}$ such that $\beta_T = \lambda g + (1-\lambda)h$, then $g\cdot \bm P_T \leq 1$, $h\cdot \bm P_T \leq 1$. But since $\beta_T \cdot \bm P_T = 1$, we must have $g\cdot \bm P_T = h\cdot \bm P_T = 1$. Thus, $g$ and $h$ are in the slice $\beta_{r_0,r_1}$. However, $\beta_T$ is extremal in this slice so $g,h \in \{\beta_T\}$.
\end{Example}
\begin{Definition}
    The \textbf{dual space} of $\mathcal{K}$, denoted $\mathcal{K}^*$, is given by:
    \begin{equation*}
        \mathcal{K}^* = \{ f \in \mathbb{R}^n : \forall x \in \mathcal{K}, \ f\cdot x \leq 1 \}
    \end{equation*}
\end{Definition}
\noindent Note that the dual $\mathcal{K}^*$ is itself a convex set over $\mathbb{R}^n$. 
\begin{Definition}
    Let $F \lhd \mathcal{K}$. The \textbf{orthogonal face} of $F$ is given by: 
    \begin{equation*}
        F^\bot = \{ f \in \mathcal{K}^* : \forall x \in F, \ f\cdot x = 1 \}
    \end{equation*}
\end{Definition}
\begin{Proposition}
    The set $F^\bot$ is a face of the convex dual $\mathcal{K}^*$.
\end{Proposition}
\begin{proof}
    Let $F \lhd \mathcal{K}$. Let $g,h \in \mathcal{K}^*$, $\lambda \in (0,1)$ such that $\lambda g + (1-\lambda) h \in F^\bot$. For all $x\in F$, $(\lambda g + (1-\lambda) h) \cdot x = 1$. Since $g,h$ are elements of the convex dual and $x \in \mathcal{K}$, their scalar product with $x$ is upper bounded by $1$. Thus $\lambda g \cdot x + (1-\lambda) h \cdot x \leq 1$. The equality implies that $g\cdot x = h \cdot x=1$ and thus $g,h \in F^\bot$.
\end{proof}
\begin{Definition}
We call \textbf{orthogonal dimension} of $F$, denoted $d_F^\bot$, the face dimension of $F^\bot$.
\end{Definition}
\begin{Definition}
We call \textbf{dimension pair} of $F$ the couple $(d_F,d^\bot_F)$.
\end{Definition}
\begin{Example}
    In the main text, we prove that $\{\bm P_T \}^\bot$ is an octagon of summits $\{S^k \cdot \beta_T\}_{k=0,...,7}$. Therefore, the orthogonal dimension of $\{\bm P_T \}$ is 2. Since the face dimension of $\{\bm P_T\}$ is 0, its dimension pair is $(0,2)$.
\end{Example}
\begin{Definition}
    We say that a face $F\lhd \mathcal{K}$ is \textbf{exposed} if there exists an element $f\in F^\bot$ such that
    \begin{equation*}
        \forall x\in \mathcal{K}, \ f\cdot x = 1 \Longrightarrow x \in F.
    \end{equation*}
\end{Definition}
\noindent An extremal point whose face $\{x\}$ is exposed is an \textit{\textbf{exposed extremal point}}.
\begin{Example}
    $\beta_T$ is exposed in $\mathcal{Q}^*$. Indeed, consider $\bm P_\star = \frac{1}{3}(\bm P_T + \bm L_{-1,-1,-1,1} + \bm L_{-1,1,1,-1})$. This point is in $\{\beta_T\}^\bot$. Moreover, recall that for any $\beta\in\mathcal{Q}^*$ and $\bm P\in\mathcal{Q}$, $\beta \cdot \bm P\leq 1$. Therefore, any $\beta$ such that $\beta \cdot \bm P_\star=1$ must satisfy $\beta \cdot \bm P_T = \beta \cdot \bm L_{-1,-1,-1,1} = \beta \cdot \bm L_{-1,1,1,-1}=1$. Since such $\beta$ is maximized by $\bm P_T$, it is in the slice $\beta_{r_0,r_1}$. But the only point in this slice giving the value $1$ on both $\bm L_{-1,-1,-1,1}$ and $\bm L_{-1,1,1,-1}$ is $\beta_T$. Hence, $\bm P_\star$ exposes $\beta_T$.
\end{Example}
\begin{Proposition}
    Let $F\lhd \mathcal{K}$. In general, we have $d_F \leq d_{F^\bot}^\bot$. If $F$ is exposed, then $d_F = d_{F^\bot}^\bot$.
\end{Proposition}
\noindent If $F$ is exposed with dimension pair $(d_F,d^\bot_F)$, the orthogonal face $F^\bot$ has opposite dimension pair $(d_{F^\bot},d^\bot_{F^\bot})=(d^\bot_F,d_F)$. In this case, we refer to $d^\bot_F$ as the \textit{\textbf{dual dimension}} of F that we denote $d^*_F$.
\begin{proof}
    By definition $d_F^\bot = d_{F^\bot}$. Now, suppose $x \in F$. For all $f\in F^\bot$, $f\cdot x=1$. Hence, $x\in (F^\bot)^\bot$ and we have $F \subset (F^\bot)^\bot$. Thus, $d_F \leq d_{F^\bot}^\bot$.
    For the second point, suppose $F$ is exposed by some $f\in \mathcal{K}^*$. Then for $x\in (F^\bot)^\bot$, we have $f\cdot x=1$. This implies $x\in F$. Thus $F=(F^\bot)^\bot$ and $d_F = d_{F^\bot}^\bot$. 
\end{proof}
\begin{Definition}
    Let $\mathcal{K}$, $\mathcal{K}'$ be two convex sets over $\mathbb{R}^n$. We say that $\mathcal{K}$ and $\mathcal{K}'$ are \textbf{isomorphic} convex sets iff there exists an isomorphism $U$ of $\mathbb{R}^n$ such that $U(\mathcal{K}) = \mathcal{K}'$.
\end{Definition}
\begin{Proposition}
     Let $\mathcal{K}$, $\mathcal{K}'$ be isomorphic convex sets under the isomorphism $U$. Then $\mathcal{K}^*$, $(\mathcal{K}')^*$ are isomorphic convex sets under the isomorphism $(U^\dagger)^{-1}$. 
\end{Proposition}
\begin{proof}
    This is only a game of definitions, as we need to prove that $(U^\dagger)^{-1} (\mathcal{K}^*) = (\mathcal{K}')^*$:
    \begin{equation*}
    \begin{split}
        (U^\dagger)^{-1} (\mathcal{K}^*) &= (U^\dagger)^{-1} \{ f \in \mathbb{R}^n | \forall x \in \mathcal{K}, \ f\cdot x \leq 1\} = \{ g \in \mathbb{R}^n | \exists f \in \mathbb{R}^n, \ g=(U^\dagger)^{-1} (f) \ \& \ \forall x \in \mathcal{K}, \ f\cdot x \leq 1\} \\
        & = \{ g \in \mathbb{R}^n | \forall x \in \mathcal{K}, \ U^\dagger (g) \cdot x \leq 1\} = \{ g \in \mathbb{R}^n | \forall x \in \mathcal{K}, \ g \cdot U(x) \leq 1\} \\
        & = \{ g \in \mathbb{R}^n | \forall x' \in \mathcal{K}', \ g \cdot x' \leq 1\}  = (\mathcal{K}')^* \\
    \end{split}
    \end{equation*}
\end{proof}
\begin{Proposition}
    Let $\mathcal{K}$, $\mathcal{K}'$ be isomorphic convex sets under the isomorphism $U$. Let $F \lhd \mathcal{K}$. Then:
    \begin{enumerate}
        \item $U(F)$ is a face of $\mathcal{K}'$,
        \item $U(F)$ has the same dimension pair as $F$.
    \end{enumerate}
\end{Proposition}
\begin{proof}
    For the first point, notice that $F\subset \mathcal{K}$ and thus $U(F) \subset \mathcal{K}'$. Now let $y',z'\in \mathcal{K}'$, $\lambda \in (0,1)$ such that $\lambda y' + (1-\lambda)z' \in U(F)$. Then there exists $y,z \in \mathcal{K}$ such that $y' = U(y)$, $z' = U(z)$. Now, by applying $U^{-1}$ and using linearity, we obtain $\lambda y + (1-\lambda)z \in F$. But $F \lhd \mathcal{K}$ so $y,z \in F$ and thus $y',z' \in U(F)$.

    For the second point, suppose that $x_1, .., x_r$ are linearly independent vectors of $F$. Then $U(x_1), .., U(x_r)$ are linearly independent vectors of $U(F)$ and thus $d_F \leq d_{U(F)}$. Now using this with $U^{-1}$, one gets $d_{U(F)} \leq d_{U^{-1}(U(F))} = d_F$ and finally $d_{U(F)} = d_F$. Now, for the orthogonal dimension, we can use the same reasoning as in the previous proof to show that $(U^\dagger)^{-1} (F^\bot) = U(F)^\bot$. We can write $d_F^\bot = d_{F^\bot} = d_{(U^\dagger)^{-1}(F^\bot)} = d_{U(F)^\bot}=d_{U(F)}^\bot$ using the first part of the proof between faces $F^\bot$ and $(U^\dagger)^{-1}(F^\bot)$.
\end{proof}

\section{SOS relaxation at level $\mathcal{T}_{1+A+B+AB+ABB'}$ for the inequality $\beta_T$} \label{sec:relaxation}
For a finite relaxation $\mathcal{T}$, we consider a generating sequence $\vec N$ of $\mathcal{A}_\mathcal{T}$. Given any symmetric matrix $W = (w_{kl})_{k,l} \in \mathcal{M}_n(\mathbb{R})$, the SOS is given by
\begin{equation}
    \vec N^\dagger \cdot W \cdot \vec N = \sum_i w_{ii} N_i^\dagger N_i + \sum_{k<l} w_{kl} \{ N_k, N_l \},
\end{equation}
where $\{ N_k, N_l \} = N_k^\dagger N_l + N_l^\dagger N_k$. The condition $1-\beta=\vec N^\dagger \cdot W \cdot \vec N$ thus imposes linear conditions on the coefficients $w_{kl}$ and the positivity constraint $W\succeq 0$ remains to be checked.

Now we focus on the finite relaxation  $\mathcal{T}_{1+A+B+AB+ABB'}$. The polynomials that nullify the state $\ket{\phi^+}$ in the ideal implementation \cref{eq:realization} at this level of relaxation are linear combinations of nine operators and the generating sequence can be chosen to be
\begin{equation} \label{eq:genseq}
    \vec N = \begin{pmatrix}
        N_0 \\
        N_1 \\
        N_2 \\
        N_3 \\
        N_4 \\
        N_5 \\
        N_6 \\
        N_7 \\
        N_8 
    \end{pmatrix} = \begin{pmatrix}
        \frac{A_0+A_1}{\sqrt{2}} - B_0 \\
        \frac{A_0-A_1}{\sqrt{2}} - B_1 \\
        1-\frac{A_0+A_1}{\sqrt{2}} B_0 \\
        1-\frac{A_0-A_1}{\sqrt{2}}B_1 \\
        \frac{A_0+A_1}{\sqrt{2}}B_1 + \frac{A_0-A_1}{\sqrt{2}}B_0 \\ 
        B_1 \left( 1 - \frac{A_0+A_1}{\sqrt{2}} B_0 \right) \\
        B_0 \left( 1 - \frac{A_0-A_1}{\sqrt{2}} B_1 \right) \\
        \left( 1 + \frac{A_0+A_1}{\sqrt{2}} B_0 \right) B_1 \\
        \left( 1 + \frac{A_0-A_1}{\sqrt{2}} B_1 \right) B_0
    \end{pmatrix} .
\end{equation}
Computing the 45 terms $N_i^\dagger N_i$, $\{ N_k, N_l \}$ for $k<l$, one can verify that the matrix
\begin{equation}
    W_3 \equiv \frac{1}{16} \begin{pmatrix}
        \sqrt{2}  & -2s & -s & & & & \\
        & 2s & 0 & & & & \\
        & & \sqrt{2} & & & & \\
        & & & 2 & 0 & s \\
        & & & & \sqrt{2}s & s\\
        & & & & & s
    \end{pmatrix},
\end{equation}
expressed in the basis $\{N_0, N_2, N_6, N_1, N_5, N_4\}$ and with $s=2-\sqrt{2}$, fulfills the condition 
\begin{equation}
    1-\beta_{(1-\frac{1}{\sqrt{2}})} = \vec N^\dagger W_3 \vec N.
\end{equation}
Furthermore, this matrix has 4 non-zero eigenvalues:
\begin{equation}
    \frac{1}{8}\left\{1-\sqrt{10-7\sqrt{2}},1+\sqrt{10-7\sqrt{2}}, 2-\sqrt{2}, \frac{3\sqrt{2}}{2}-1\right\}.
\end{equation} 
Since all 4 are positive, $W_3 \succeq 0$, and $W_3$ provides a certificate for the inequality $\beta_{(1-\frac{1}{\sqrt{2}},0)}:=\beta_T$.

\section{Order two variations}\label{sec:orderTwo}
The second order variations of the function $\beta \cdot \bm P_{\theta,a_x,b_y}$ can be used to rule out some inequalities of the plane $\beta_{r_0,r_1}$. Indeed, for the value $1$ of the inequality to be a true maximum at the Tsirelson realization $\bm P_T$, the hessian matrix must be negative. Now, looking at second order variations, one must be a bit careful when it comes to which points are close to the Tsirelson realization: while it can be obtained with parameters $(\pi/4, 0,\pi/2, \pi/4, -\pi/4)$, it can also be achieved by any local rotations of both parties simultaneously - $s_\alpha = (\pi/4, \alpha, \pi/2+\alpha, \pi/4 + \alpha, -\pi/4 + \alpha) $ for any $\alpha \in [0,2\pi)$. As such, the hessian matrix should be considered for any possible value $\alpha$:
\begin{equation*}
    H_{r_0,r_1,\alpha} = \left( \beta_{r_0,r_1} \cdot \left. \frac{\partial^2 \bm P_{\theta, a_x, b_y}}{\partial_i \partial_j } \right|_{s_\alpha} \right)_{i,j}
\end{equation*}
where the partial derivatives on $i,j$ run across all five parameters $\theta, a_0, a_1, b_0, b_1$. This matrix can be computed explicitly:
\begin{equation*}
    H_{r,\gamma,\alpha} = \begin{pmatrix}
        -2 & & & & \\
       2rs_\alpha s_{\gamma+\pi/4} & -1/2 & & &  \\
       -2rc_\alpha s_{-\gamma+\pi/4} & 0 & -1/2 & & \\
       -2rs_\gamma s_{\alpha+\pi/4} & 1/4 & 1/4 & -1/2 &  \\
       2rc_\gamma s_{-\alpha+\pi/4} & 1/4 & 1/4 & 0 & -1/2 
    \end{pmatrix}
\end{equation*}
where we used the polar parametrization of the plane $(r_0,r_1) = (rc_\gamma,rs_\gamma)$ and all unwritten terms are given by symmetry of the hessian matrix. Now since we know that the point $r=0$ is a true maxima and that the region of interest in the plan is convex, we can express the problem of excluding inequalities as the following SDP: 
\begin{equation*}
    \begin{split}
        \max& \quad r \\
        & \text{s.t } H_{r,\gamma,\alpha} \preceq 0 
    \end{split}
\end{equation*}
This optimization can be performed numerically. The result does not depend on the parameter $\gamma$ and is always $r_{\max} = 0.5$. The corresponding disk is represented in \cref{fig:appendix}.

\begin{figure}
    \centering
    \includegraphics[width=0.6\textwidth]{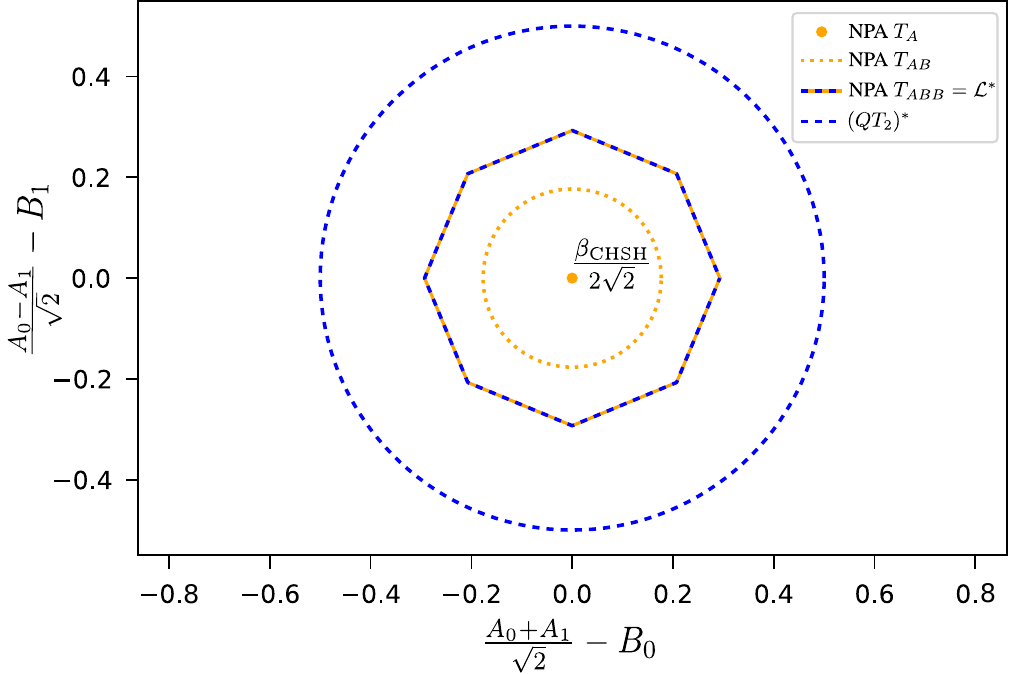}
    \caption{This figure summarizes the bounds on the dual set $\mathcal{Q}^*$ in the slice considered the main text. The blue lines are outer relaxations of the dual in the plane of inequalities $\beta_{r_0,r_1}$ while the orange lines are inner certification of the dual in this plane, obtained by the NPA hierarchy at different levels. The dotted blue circle correspond to the maximal radius for which second order variations of two-dimensional realizations doesn't increase the inequality value: every inequality outside are not in $\mathcal{Q}^*$ as small variations around $\bm P_{T}$ can grant values of $\beta_{r_0,r_1}$ larger than $1$. The dashed blue octagon correspond to the region of inequalities having a local bound of $1$: all inequalities outside this region are not in $\mathcal{Q}^*$ as they are not in $\mathcal{L}^*$. On the other hand, the orange dot in the middle, the dotted orange circle and the dashed orange octagon correspond to certifications of quantum bound equal to $1$ using the NPA hierarchy for the relaxations $T_{A}=\mathcal{T}_{1+A+B}$, $T_{AB}=\mathcal{T}_{1+A+B+AB}$ and $T_{ABB}=\mathcal{T}_{1+A+B+AB+ABB'+AA'B}$ respectively. Note that the bound for the relaxation level $\mathcal{T}_{AB}$ was obtained numerically. The orange and blue octagon coincide, meaning that the set $\mathcal{Q}^*$ in this slice is exactly the region inside the octagon.}
    \label{fig:appendix}
\end{figure}

\section{The extremal points on the face $\beta_T$}\label{sec:noOtherPoint}
The extremal Tsirelson inequality $\beta_T$ forms part of the boundary of the quantum set $\mathcal{Q}$. In turn, all the quantum behaviors $\bm P$ which saturate this inequality, i.e. for which $\beta_T\cdot \bm P=1$, define a face of the quantum set, which is a convex set in itself. In this section, we find all the extremal points of this face.

Due to \cite{Mikosnuszkiewicz23}, the search of extremal points can be reduced to the points $\bm P_{\theta,a_x,b_y}$ achieved with a partially entangled two qubits state $\ket{\phi_\theta}=c_\theta\ket{00}+s_\theta\ket{11}$, $\theta\in[0,\pi/4]$, and real measurements in the $Z-X$ plane. Two cases can be treated directly. The first one is $\theta = \pi/4$. This implies that all marginal terms $\langle A_x \rangle$, $\langle B_y \rangle$ are $0$ and thus $\langle \beta_\text{CHSH} \rangle = 2\sqrt{2}$ and the only possible realization is $\bm P_T$. The second case is $\theta =0$ as this implies that the point is local. One can check that there are only two extremal local points reaching the value $1$ for $\beta_T$. We now prove that there are no other extremal points on this face.

For this, we consider all possible quantum realization giving $\beta_T\cdot \bm P_{\theta,a_x,b_y}=1$. We have the SOS decomposition $1-\beta_T = N^\dagger W_3 N = \sum_{k=0}^8 w_k O_k^\dagger O_k$, where $w_k$ are the eigenvalues of $W_3$ and $O_k$ are the elements of the relaxation whose coefficients in terms of the generating sequence \cref{eq:genseq} are given by the coefficients of the eigenvectors of $W_3$ in the canonical basis. Therefore, when value $1$ is reach for the inequalities, all operators $\hat O_k$ associated to a non-zero eigenvalue should nullify the measured state: 
\begin{equation*}
    w_k \neq 0\  \Longrightarrow \ \Hat O_k \ket{\phi_\theta} = 0 
\end{equation*}
Since $W_3$ has 4 non-zero eigenvalues, we obtain 4 equations. Using linear combinations of these equations, one can derive $(\hat N_0 - \hat N_2)\ket{\phi_\theta}=0$. Projecting this equation onto all subspaces $\ket{ij}$ gives:
\begin{subequations}
\begin{equation} \label{eq:eq1}
    c_\theta \frac{c_{a_0}+c_{a_1}}{\sqrt{2}} - c_\theta c_{b_0} = c_\theta - \frac{1}{\sqrt{2}}\sum_x (c_\theta c_{a_x}c_{b_0} + s_\theta s_{a_x}s_{b_0})
\end{equation}
\begin{equation} \label{eq:eq2}
    - s_\theta \frac{c_{a_0}+c_{a_1}}{\sqrt{2}} + s_\theta c_{b_0} = s_\theta - \frac{1}{\sqrt{2}}\sum_x (s_\theta c_{a_x}c_{b_0} + c_\theta s_{a_x}s_{b_0})
\end{equation}
\begin{equation} \label{eq:eq3}
    s_\theta \frac{s_{a_0}+s_{a_1}}{\sqrt{2}} - c_\theta s_{b_0} = - \frac{1}{\sqrt{2}}\sum_x (c_\theta c_{a_x}s_{b_0} - s_\theta s_{a_x}c_{b_0}) 
\end{equation}
\begin{equation} \label{eq:eq4}
    c_\theta \frac{s_{a_0}+s_{a_1}}{\sqrt{2}} - s_\theta s_{b_0} = - \frac{1}{\sqrt{2}}\sum_x (-s_\theta c_{a_x}s_{b_0} + c_\theta s_{a_x}c_{b_0}) 
\end{equation}
\end{subequations}
Now a linear combination of \cref{eq:eq3} and \cref{eq:eq4} gives
\begin{equation} \label{eq:nullifprojected}
    c_{2\theta} s_{b_0} = - \frac{1}{\sqrt{2}}\sum_x (-c_{a_x}s_{b_0} + s_{2\theta} s_{a_x}c_{b_0}).
\end{equation}
Furthermore, we can use the fact that the value of $\beta_T\cdot \bm P_{\theta,a_x,b_y}$ is maximal to ensure that first order variations of the measurements does not increase the value at first order. In particular: 
\begin{equation}
    0 = \beta_T \cdot \frac{\partial P_{\theta,a_x,b_y}}{\partial b_0} = - (1-\frac{1}{\sqrt{2}})c_{2\theta}s_{b_0} + \frac{1}{2\sqrt{2}} \sum_x (-c_{a_x}s_{b_0} + s_{2\theta} s_{a_x}c_{b_0})
\end{equation}
The two previous equations are only compatible when $c_{2\theta} s_{b_0} = 0$. Therefore, we have either $b_0 = 0$ or $b_0 = \pi$. In the first case, \cref{eq:eq1} implies $c_{a_0}+c_{a_1} = \sqrt{2}$ and in the second case, \cref{eq:eq2} implies $c_{a_0}+c_{a_1} = -\sqrt{2}$. In all cases $\left( \frac{A_0+A_1}{\sqrt{2}}-B_0\right) \cdot \bm P_{\theta,a_x,b_y} = 0$. Thus $\beta_T \cdot P_{\theta,a_x,b_y} = 1$ implies $\beta_\text{CHSH} \cdot P_{\theta,a_x,b_y} = 2\sqrt{2}$ and the only compatible point is $\bm P_T$.

\end{document}